\documentclass[aps,prl,groupedaddress,twocolumn,showpacs,superscriptaddress,
nofootinbib]{revtex4-1}

\usepackage{pifont}
\usepackage{palatino}
\usepackage{diagbox}
\usepackage{graphics,epsfig,subfigure,dsfont,amssymb,amsthm,amsfonts,amsbsy,mathrsfs,amscd}
\usepackage{epsfig}
\usepackage{color}
\usepackage{physics}
\usepackage{framed}
\usepackage[T1]{fontenc}

\newtheorem{theorem}{Theorem}

\newtheorem{observation}[theorem]{Observation}

\newtheorem{Lemma}[theorem]{Lemma}

\def\be{\begin{equation}}
\def\ee{\end{equation}}
\def\ba{\begin{array}}
\def\ea{\end{array}}
\def\tr{\mathrm{tr}}

\def\k{{\bf k}}

\newcommand{\bea}{\begin{eqnarray}}
\newcommand{\eea}{\end{eqnarray}}

\def\bi{\begin{itemize}}
\def\ei{\end{itemize}}
\def\bc{\begin{center}}
\def\ec{\end{center}}

\def\C{\hbox{$\mit I$\kern-.7em$\mit C$}}
\def\R{\hbox{$\mit I$\kern-.6em$\mit R$}}
\def\N{\hbox{$\mit I$\kern-.6em$\mit N$}}

\def\ket#1{|#1\rangle}
\newcommand{\one}{\mbox{$1 \hspace{-1.0mm}  {\bf l}$}}
\def\tr{\mathrm{tr}}
\def\ket#1{\left| #1\right>}
\def\bra#1{\left< #1\right|}

\newcommand{\proj}[1]{\ket{#1}\bra{#1}}
\newcommand{\kb}[2]{\ket{#1}\bra{#2}}

\usepackage{hyperref}
\hypersetup{colorlinks=true, linkcolor=blue, citecolor=blue, urlcolor=blue}

\begin{document}

\title{Structure of dimension-bounded temporal correlations}

\author{Yuanyuan Mao}
\affiliation{Naturwissenschaftlich-Technische Fakult{\"a}t, Universit{\"a}t Siegen, Walter-Flex-Stra\ss e 3, 57068 Siegen, Germany}
\author{Cornelia Spee}
\affiliation{Institute for Quantum Optics and Quantum Information (IQOQI), Austrian Academy of Sciences, Boltzmanngasse 3, 1090 Vienna, Austria}
\affiliation{Naturwissenschaftlich-Technische Fakult{\"a}t, Universit{\"a}t Siegen, Walter-Flex-Stra\ss e 3, 57068 Siegen, Germany}
\author{Zhen-Peng Xu}
\affiliation{Naturwissenschaftlich-Technische Fakult{\"a}t, Universit{\"a}t Siegen, Walter-Flex-Stra\ss e 3, 57068 Siegen, Germany}
\author{Otfried G{\"u}hne}

\affiliation{Naturwissenschaftlich-Technische Fakult{\"a}t, Universit{\"a}t Siegen, Walter-Flex-Stra\ss e 3, 57068 Siegen, Germany}

\begin{abstract}

We analyze the structure of the space of temporal correlations 
generated by quantum systems. We show that the temporal correlation 
space under dimension constraints can be nonconvex. For the general 
case, we provide the necessary and sufficient dimension of a quantum 
system needed to generate a convex correlation space for a given scenario. 
We further prove that this dimension coincides with the dimension 
necessary to generate any point in the temporal correlation polytope. 
As an application of our results, we derive nonlinear inequalities 
to witness the nonconvexity for qubits and qutrits in the simplest 
scenario, and present an algorithm which can help to find the minimum 
for a certain type of nonlinear expressions under dimension constraints. 
\end{abstract}

\maketitle

\textit{Introduction.---} 
States of a quantum system are mathematically described by vectors in a 
Hilbert space. When no {\it a priori} information about the measurements 
or the states is known, one of the intrinsic properties we can possibly 
tell about an unknown quantum system is the dimension of its underlying 
Hilbert space. The dimension is considered as a valuable resource from 
an information-theoretical viewpoint \cite{Hol73,LBA09,AGM06}. 
Higher-dimensional quantum systems have been proven to be able to perform 
better in some tasks like quantum key distribution \cite{BPT00, CBK02} 
and so they can be used to implement more powerful protocols than 
lower-dimensional quantum systems \cite{EFK18}. 

But what can be concluded, if the dimension is limited? For instance, 
in the semi-device-independent framework of quantum information 
processing, nothing else but the dimension of the quantum system is 
assumed \cite{PaB11,LVB11}. The system is then measured in different 
experimental configurations and the statistics of the outcomes, usually 
referred to as quantum correlations, are recorded. The typical example
of this scenario is the Bell test which proves that quantum mechanics is 
nonlocal. The resulting spatial correlations play a central role in many 
quantum information protocols, such as quantum key distribution and randomness 
certification \cite{PaB11,LPY12}.  Preceding works studied the space 
of quantum correlations arising from quantum systems of different 
dimensions in many scenarios \cite{BPA08,BNV13,GBH10,NFA15,BBP15}, 
with various techniques using convex optimization designed to find 
the bound of some linear functionals of the correlations achievable 
with a given dimension \cite{NaV15,TRR19}. In the Bell scenario, 
however, the sets of correlations  arising from dimension-bounded Hilbert 
spaces are typically nonconvex \cite{SVW16,DoW15,BQB14}. Hence what 
linear functionals characterize  are essentially the convex hulls 
of correlations sets, rather than correlation sets themselves. Besides, 
some of these Bell-type dimension tests have recently been critically 
investigated, as they may not characterize the experimentally relevant 
figures of merit \cite{CCB17,KRB18}.

\begin{figure}[t]
\includegraphics[width=0.9\columnwidth]{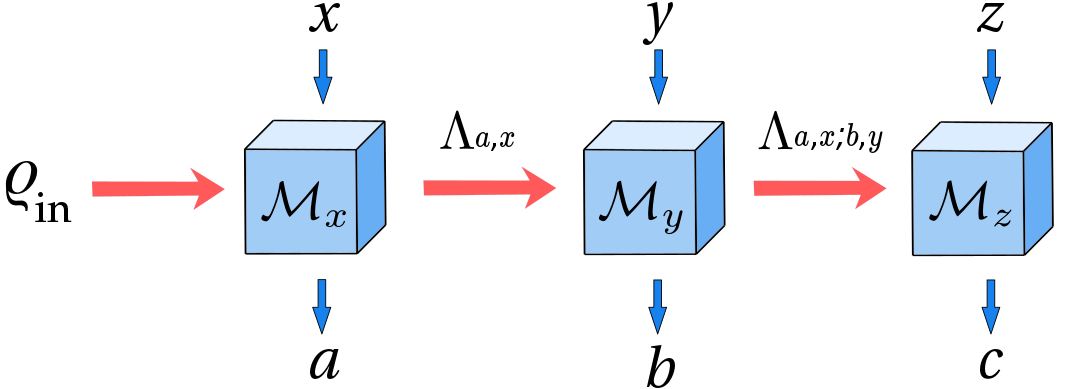}
\caption{A quantum system with initial state $\rho_{\rm in}$ is measured several times, the measurements can be repeated. The output state after each measurement will be subjected to a quantum dynamics which may depend on the prior choice of measurements and the outcome of the measurements. In this figure we depict this scenario for $L=3$.}
\label{fig1}
\end{figure}

In this paper we consider a different model, where measurements are performed in a temporal sequence \cite{BKM15,BMK13,Zuk14,HSG18,SSG18}, instead 
of spatial correlations investigated in Bell tests. The resulting temporal correlations can be
used to violate Leggett-Garg inequalities \cite{lg1,lg2}, proving quantum mechanics is not a theory of macroscopic realism.
We study the structure of temporal quantum correlations generated by dimension-bounded systems. First we will prove that already for the simplest scenario, the correlation spaces obtained by qubit or qutrit systems are nonconvex, and we provide nonlinear witnesses detecting this nonconvexity. Namely, they can distinguish quantum systems with different dimensions even if the convex hulls of the correlation spaces are the same. For general scenarios, we give a formula for the necessary and sufficient dimension of quantum systems, from which a convex set of temporal correlations can be obtained. As an application, we show that our nonconvexity witnesses are also qualitatively better dimension witnesses than linear ones. In order to derive nonlinear inequalities able to test higher dimensions, we present an iterative algorithm which allows us optimize a certain type of temporal correlation polynomials over dimension-bounded Hilbert spaces.

\textit{The space of temporal correlations.---} 
As illustrated in Fig. \ref{fig1}, a single system prepared in an initial 
state $\rho_{\rm in}$ is subjected to a sequence of measurements of certain length $L$. At each time step, a measurement selected from a given set 
$\{\mathcal{M}_0,\mathcal{M}_1,\ldots,\mathcal{M}_{S-1}\}$ is performed according to the input from an input alphabet $\mathcal{X}=\{0,1,\ldots,S-1\}$, and after each measurement an output from an alphabet $\mathcal{A}=\{0,1,\ldots,O-1\}$ is obtained. No assumption on the type of measurements will be imposed. In between two measurements we allow for an arbitrary quantum dynamics,
which may depend on the former choice of measurements and the measurement 
outcomes.  Given an initial state $\rho_{\rm in}$, one obtains a probability distribution $p(ab\cdots|xy\cdots)$ for any input sequence $xy\cdots$. We call the collection of the probability distributions generated by all possible inputs a temporal correlation. As a result of causality, the 
choice of latter measurements can not affect the outcomes of former measurements. Hence, the temporal correlations have to fulfill the arrow of time
(AoT) constraints \cite{ClK16}. For a two-step process, the constraints read
\begin{equation}
\sum_b p(ab|xy) =\sum_b p(ab|xy'), 
\end{equation}
for all $a,b \in \mathcal{A}, ~ x,y,y' \in \mathcal{X}.$
If there is no further assumption on the dimension of the quantum system, 

for any given $L$, $S$, and $O$, the temporal correlations form a polytope 
denoted by $P^L_{S,O}$ \cite{ClK16}. The extreme points of this polytope 
are the deterministic assignments, where each measurement has a fixed 
outcome and the AoT constraints are fulfilled \cite{AGC16,HSG18}. 
A correlation $\{p(abc\ldots|xyz\ldots)\}$ is in the temporal correlation 

polytope if and only if it can be decomposed 
as
\begin{equation}
p(abc\ldots|xyz\ldots)=p(a|x)p(b|a,xy)p(c|ab,xyz)\ldots,
\label{decomp}
\end{equation}
with $p(a|x),p(b|a,xy),p(c|ab,xyz),\ldots$ denoting the local probability 
distribution where the measurement choice and their outcomes in the preceding time steps are fixed \cite{HSG18}. It has been shown that any correlation obeying the
AoT condition can be reached in quantum mechanics \cite{Fri10,HSG18}, in contrast to the non-signaling polytope in the Bell scenario \cite{PoR94}, 
where not all the points can be realized.

\textit{Nonconvexity in the simplest case.---} 
The most basic experimental setup is to measure an uncharacterized quantum system twice, producing binary strings $ab \in \{0,1\}^{\otimes 2}$. The performed measurements are chosen from a set of two-outcome measurements $\{\mathcal{M}_0,\mathcal{M}_1\}$, based on the input string $xy\in \{0,1\}^{\otimes 2}$. Qubits can already be distinguished from higher-dimensional systems with this simple setup, since one can reach all the extreme 
points of the polytope by using qutrits, but not qubits \cite{HSG18}. Moreover, as we prove below, the set of quantum correlations generated by a qubit is not convex. For example, the two extreme points of the correlation polytope
\begin{equation}
\begin{split}
p_1& : p(10|00)=p(10|01)=p(01|10)=p(00|11)=1,\\
p_2& : p(10|00)=p(10|01)=p(10|10)=p(10|11)=1,
\end{split}
\end{equation}
can be attained by measuring a single qubit \cite{HSG18}. Nevertheless, the mixture of both, $p_m = \frac{p_1+p_2}{2}$, can not be achieved by a 
qubit. 
This can be seen as follows: In order to realize the correlation $p_m$, both measurements $\mathcal{M}_0$  and $\mathcal{M}_1$ have to be able to give each of the two results. Moreover, measuring $\mathcal{M}_0$ in the first step gives result "1" with certainty and in the second step if $\mathcal{M}_0$ was measured in the first step, it produces result "0" with certainty. This means both of its effects have to be projective operators. 
Without loss of generality, we denote the initial state by $|1\rangle$. Then the measurement $\mathcal{M}_0$ is measuring the observable $\sigma_z$, and the intermediate state after choosing $\mathcal{M}_0$ as first measurement is precisely $|0\rangle$. Based on the observation that measuring $\mathcal{M}_1$ on state $|0\rangle$ always gives outcome "0", we can tell that the effect of $\mathcal{M}_1$ corresponding to outcome "0" is of 
the form $|0\rangle\langle 0|+\epsilon |1\rangle\langle 1|$, with $\epsilon \in [0,1)$. If we measure $\mathcal{M}_1$ twice, the second step will 
give outcome "0" with certainty, which indicates that the intermediate state after measuring $\mathcal{M}_1$ is also the $|0\rangle$. However, in this case the probability $p(01|10)$ vanishes, which contradicts $p(01|10)=1/2$.

Besides case to case analysis, the nonconvexity can also be detected by nonlinear inequalities:
\begin{observation}
For correlations resulting from arbitrary measurements on a qubit, it holds that
\begin{equation}
\begin{split}
\mathcal{S}_1=&2p(0|0)+p(0|0,00)+2p(0|1)\\
+&p(0|0,11)+p(1|0,10)p(1|0,01) \underset{d=2}{\leq} 6.
\end{split}
\label{eq1}
\end{equation}
\end{observation}

Here $p(b|a,xy)=p(ab|xy)/p(a|x)$ denotes the probability of obtaining the outcome "b" when measuring the measurement $\mathcal{M}_y$ in the second time step, given that the measurement $\mathcal{M}_x$ was measured in the first time step, and outcome "a" was obtained. The proof of Eq.~(\ref{eq1}) is presented in the Appendix A, wherein also an example of non-convexity detected by Eq.~(\ref{eq1}) is given. In this example, both extreme points we consider are achievable by a qubit, but the uniform mixture of them violates the inequality as demonstrated in Fig.~\ref{fig2}. The maximal value $\mathcal{S}_1=7$ can be achieved by an extreme point of the polytope, which
corresponds to a qutrit system \cite{HSG18}.

\begin{figure}
\centering     
\includegraphics[width=0.8\columnwidth]{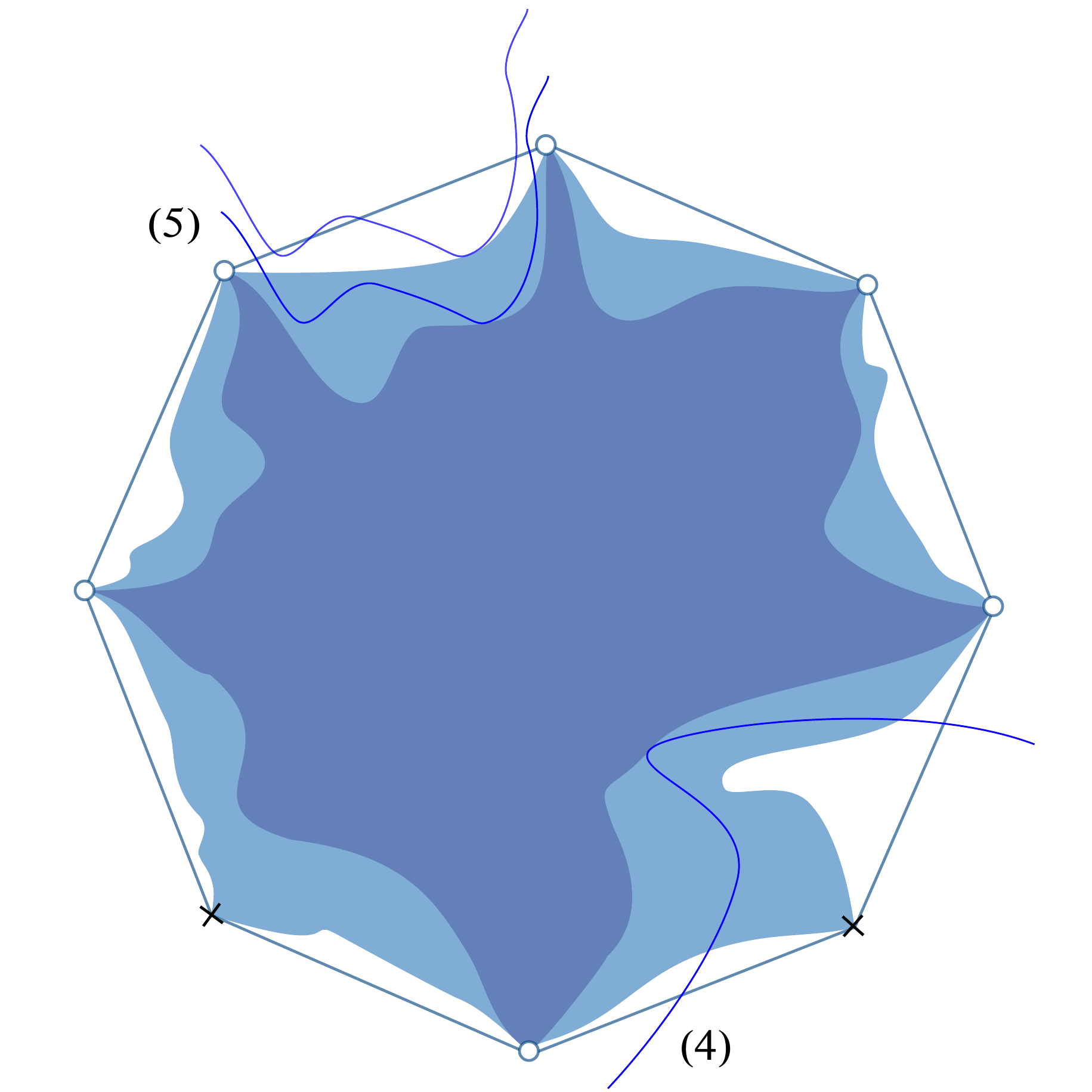}
\caption{Schematic illustration of the temporal correlation space in the simplest case. The octagon denotes the temporal correlation polytope, the 
darker area is the temporal correlation space generated by a qubit, and the lighter area denotes the temporal correlations that can be reached by a qutrit, but not a qubit. We label the extreme points achievable by a qubit with circles and other extreme points with crosses. The curve in the bottom describes (\ref{eq1}), whose maximum is achieved by an extreme point which can not realized by qubit. The algebraic maximum of inequality (\ref{eq2}), described by the double curves at the upper left, is achieved 
by the uniform mixture of two extreme points that are achievable by qubits.} 
\label{fig2}
\end{figure}

In the simplest scenario $L=S=O=2$ all the extreme points are already 
achievable by qutrits, so linear dimension witnesses could not distinguish 
qutrits from higher-dimensional quantum systems. Still, nonlinear
criteria can do that, as the following inequality shows:

\begin{observation}
For arbitrary measurements we have that
\begin{align}
\mathcal{S}_2=&p(0|0,00)+p(0|0,01)+p(0|0,10)+p(1|1,00)
\label{eq2}
\\
&+p(1|1,10)+p(1|1,11) \nonumber \\
&+p(1|0,11)+p(0|1,01) \underset{d=2}{\leq} 4+2\sqrt{2} \underset{d=3}{\leq} 5+\sqrt{5},
\nonumber
\end{align}
where the first bound holds for a qubit, and the second bound for a qutrit. The 
algebraic maximum $\mathcal{S}_2=8$ can be reached by a four-level system.
\end{observation}

It should be noted that the above inequality can also be interpreted in the 
prepare-and-measure scenario where the pair $(a,x)$ determine the prepared 
state and $y$ the input defines the measurement setting. In this context, 
it corresponds to a quantum random access code \cite{WCD08}, for which the qubit bound has already been shown analytically \cite{ANT02}, and the qutrit bound has been obtained numerically \cite{NFA15}. This connection allows one to use inequalities and techniques known in the prepare and measure scenario for the study of temporal correlations and vice versa.
In Appendix B we provide a proof of the Observation, in particular we prove the qutrit bound analytically. Alongside we show an example of two extreme points, who both can be reached by measuring a qubit, but the uniform mixture requires a four-level system and reaches $\mathcal{S}_2=8$. 

\textit{Mixing with white noise.---} 
We will consider in the following that the experiment is affected by noise. We call the noise local white noise if the experiment is only disturbed at one time step, which causes the local-in-time distribution $\{p(a|hx)\}$ to be mixed with a local uniform distribution $\{p(a|hx)=\frac{1}{O}\}$. Here $h$ stands for the history, i.e., the chosen measurements and 
their outcomes before the time step. If the correlation $\{p(abc\cdots|xyz\cdots)\}$ itself is mixed with a uniform distribution $\{p(abc\cdots|xyz\cdots)=\frac{1}{O^L}\}$, we say that the noise is a global white noise. Counterintuitively, mixing a correlation $\{p(abc\cdots|xyz\cdots)\}$ with local or global white noise does not necessarily reduce the dimension required to realize it. This also exemplifies the nonconvexity of dimension-bounded temporal correlations.  Here we discuss the two kinds of white noise separately.

(i) Local white noise. For example, if the correlation is affected by local white noise to step two, the conditional probability distribution at the second step $\{p(b|a,xy)\}$ for chosen $a,x,y$ is mixed with $\{p(b|a,xy)=\frac{1}{O}\}$.  Obviously for certain correlations, this process can have more outcomes for one time step, which may increase the necessary 
dimension of quantum system.

(ii) Global white noise. Consider a given correlation $\{p(abc\cdots|xyz\cdots)\}$ is mixed with the identity correlation $\{p(abc\cdots|xyz\cdots)=\frac{1}{O^L}\}$. Here we present two examples, where the necessary dimension increases.

{\it Example 1.}  Consider a trivial extreme point in the (2-2-2) scenario,
$p(00|00)=p(00|01)=p(00|10)=p(00|11)=1.$
Its uniform mixture with the identity is
\begin{equation}
p(ab|xy)=
\begin{cases}
\frac{5}{8}, ~~~\text{for}~~a=0, b=0,\\
\frac{1}{8}, ~~~\text{otherwise},
\end{cases}
\end{equation}
which cannot be generated by a one-dimensional quantum system in contrast 
to the original correlations.

{\it Example 2.} Consider the extreme point defined by $p(00|00)=p(00|01)=p(00|10)=p(01|11)=1$. It can be easily seen that this point can be realized with measurements on a qubit \cite{HSG18,SBG19}. However, as we will show in Appendix C, the convex combination of this point and sufficiently weak global white noise requires at 
least a qutrit for its realization.

>From the discussion above, we see that the correlation space expands while the dimension $d$ of the underlying quantum system increases, until the whole correlation polytope is obtained. For the simplest scenario, the 
nonconvexity of the qutrit correlation space shows that the whole correlation polytope can not be reached with a qutrit, although all the extreme points can be achieved. A natural question then arises: which dimension is needed in order to obtain the entire temporal correlation polytope? We give an explicit formula for this dimension in the following, and we show 
that any correlation space generated by a system with smaller dimension is nonconvex.

\textit{General scenarios.---} 
For an arbitrary given scenario with $L$ measurement steps, $S$ possible measurements, and $O$ possible outcomes per measurement, the temporal correlation polytope $P^L_{S,O}$ has $(O^S)^{\frac{S^L-1}{S-1}}$ extreme points \cite{HSG18}. The following theorem provides the smallest dimension of a quantum system, such that the generated set of temporal correlations will be convex. We call this the \textit{critical dimension} $\mathcal{D}(L,S,O)$. We show moreover that the set of temporal correlations generated 
by a quantum system of critical dimension is already the temporal correlation polytope $P^L_{S,O}$. Hence, the set of temporal correlations of a quantum system cannot be extended by increasing its dimension beyond the critical dimension.

\begin{theorem}
The critical dimension is given by the following formula
\begin{equation}
\mathcal{D}(L,S,O)=\min \{O^S,\frac{(OS)^{L}-1}{OS-1}\}.
\end{equation}
Quantum systems with a dimension larger than or equal to the critical 
dimension generate the correlation polytope $P^L_{S,O}$. Moreover, 
any correlation space generated by quantum systems with smaller 
dimension is nonconvex.
\label{thm}
\end{theorem}

To give an example, with this formula we can calculate the critical dimension of the simplest case as $\mathcal{D}(2,2,2)=4$. The detailed proof 
is presented in Appendix D.
A sketch of the proof is as follows: In order to show that the critical dimension is necessary to achieve all the correlations in the polytope, we 
consider two density matrices which have to be able to each realize a certain local-in-time correlation. Then we show an upper bound on the overlap of the eigenstates corresponding to the maximal eigenvalue of these two 
density matrices. One can then show that if the pairwise upper bound is low enough for a set of states, these states have to be linearly independent, which proves the necessity of the critical dimension. For the other direction, we construct protocols to realize an arbitrary point in the correlation space with a $\mathcal{D}$-dimensional system. Then we give examples contradicting the convexity of correlation space generated by systems whose dimension is smaller than critical dimension. Our results can be also straightforwardly used in the prepare-and-measure scenario where in addition to constraints on the dimension among others also the minimal 
overlap assumption has been considered \cite{SCB19}.

\textit{Numerical algorithms.---} 
Finally, let us provide a see-saw algorithm that can find the maximum of general polynomial, if the maximum is attained on pure states and projective measurements under dimension constraints. The polynomials discussed in this paper all fulfill this assumption. Exploiting the correspondence between length-two temporal correlations and the prepare and measure setup, our method can be utilized in both scenarios.

Consider any given polynomial $p(X_1,X_2,\ldots,X_n)$ where the $X_i$ are 
the involved probabilities of the form $p(a|x)$ or $p(b|a,xy)$. Since every maximization problem can be converted into a minimization problem, we only present the method for finding the minimum of such a polynomial. To find the minimum of $p(X_1,X_2,\ldots,X_n)$ for a $d$-dimensional quantum 
system, we can first choose a random number $q$, and check whether $p(X_1,X_2,\ldots,X_n)$ can achieve a value smaller than $q$ with correlations obtained from measuring a $d$-dimensional system. We illustrate this using the $d=2$ case as an example. For a correlation that can be produced by a qubit, its corresponding $(X_1,X_2,\ldots,X_n)$ has a quantum representation $X_i=\tr (\rho_i M_i)$, with $\rho_i$ being the initial or intermediate states and $M_i$ the measurement effects. By assumption, the polynomial is minimized by a correlation with pure states $\rho_i = |\psi_i\rangle\langle\psi_i|$ and projective measurement effects $M_i = |\phi_
i\rangle\langle\phi_i|$. For this correlation we can construct a $2\times 
2n$ matrix
\begin{equation}
\Gamma =  \begin{pmatrix}
   |\psi_1\rangle, \ldots,    |\psi_n\rangle,   |\phi_1\rangle,  \ldots , 
|\phi_n\rangle|
\end{pmatrix}.
\label{eq5}
\end{equation}
Then, the matrix $\Gamma^\dagger \Gamma$ is a $2n \times 2n$ positive semi-definite matrix with all diagonal entries equal to $1$ and rank 2. Every $X_i =\tr (\rho_i M_i)=|\langle\psi_i|\phi_i\rangle|^2$ is the absolute square of a certain entry. If the minimum of $p(X_1,X_2,\ldots,X_n)$ 
is smaller than a number $q$, then there should exist a common object in the following two sets of $2n \times 2n$ matrices:

\noindent
($M_1$) Rank two positive semi-definite matrices.

\noindent
($M_2$) Hermitian matrices  with the main diagonal $(1,1,\ldots,1)$, whose entries corresponding to $\{X_i\}$ satisfy the inequality
\begin{equation}
p(X_1,X_2,\ldots,X_n)\leq q.
\end{equation}

To examine the existence of such a matrix, one can iterate between these two sets.
Starting from a matrix in $M_1$ one can find analytically the closest matrix
in $M_2$. For this matrix, one can then find analytically the closest
matrix in $M_1$ again, etc. We describe the algorithm in detail in Appendix E. 
A common object exists if the iteration converges, the converse is however 
not true. In Appendix F we give an example of applying our method to treat 
the inequality (\ref{eq2}) numerically.

\textit{Conclusions.---} 
We characterized the nonconvex structure of temporal correlation space generated by finite-dimensional quantum systems.  For arbitrary scenarios, we derived the critical dimension of quantum systems to generate a convex 
set of temporal correlations. We established nonlinear inequalities for the simplest case with upper bounds satisfied by qubits or qutrits respectively. These nonlinear inequalities can serve as implementable dimension witnesses. In this way, our results might trigger experimental investigations of the performance of systems with different finite dimensions.

Note that our setting allows for arbitrary dynamics happening between adjacent time steps. The structure of the temporal correlation space can change if we limit the possible intermediate channels to certain classes, e.g., Markovian channels. It would be interesting to study the features of correlation space corresponding to restricted quantum channels. This might inspire a general method to experimentally reveal the properties of quantum channels by analyzing the obtained temporal correlations. We leave this problem for future research.

We would like to thank Marco T\'ulio Quintino for discussions. We acknowledge 
financial support by the Deutsche Forschungsgemeinschaft (DFG, German Research
Foundation, project numbers 447948357 and 440958198), the Sino-German Center 
for Research Promotion (Project M-0294), the ERC (Consolidator Grant 
683107/TempoQ), and the DAAD (Projekt-ID: 57445566). Y. M. acknowledges funding 
from a CSC-DAAD scholarship. C. S. acknowledges support by the Austrian 
Science Fund (FWF): J 4258-N27. Z.-P. X. thanks the support of the Alexander 
von Humboldt Foundation.

\renewcommand{\theequation}{S\arabic{equation}}
\renewcommand{\thefigure}{S\arabic{figure}}


\onecolumngrid

\section{Appendix A: Proof of Observation 1}
Here we prove Observation 1, i.e., we show that for a qubit the quantity
\begin{equation}
\mathcal{S}_1:=2p(0|0)+p(0|0,00)+2p(0|1)+p(0|0,11)+p(1|0,10)p(1|0,01)
\end{equation}
can not exceed $6$. In order to obtain the upper bound, we first parametrize the measurement effects $\mathcal{E}_{r|s}$ corresponding to the outcome $r=0,1$ of measurement $M_s, s=0,1$ by
\begin{equation}
\begin{split}
\mathcal{E}_{0|0} &= p_0 |a_0\rangle\langle a_0| +q_0 |a_0^\bot\rangle\langle a_0^\bot |= \frac{p_0+q_0}{2} \one+\frac{p_0-q_0}{2}\vec{c}\cdot\vec{\sigma},~~~\mathcal{E}_{1|0} = \one-\mathcal{E}_{0|0}, \\
\mathcal{E}_{0|1} &= p_1 |a_1\rangle\langle a_1| +q_1 |a_1^\bot\rangle\langle a_1^\bot |= \frac{p_1+q_1}{2} \one+\frac{p_1-q_1}{2}\vec{d}\cdot\vec{\sigma},~~~\mathcal{E}_{1|1} = \one-\mathcal{E}_{0|1},
\end{split}
\end{equation}
where $\one$ is the identity, $\vec{\sigma}$ is the matrix vector of Pauli matrices, the real vectors $\vec{c}$ and $\vec{d}$ are of unit length, and $p_0, p_1, q_0, q_1 \in [0,1]$. Moreover, we denote the initial state by $\rho_{\rm in}$ and the the post-measurement states corresponding to the effects $\mathcal{E}_{0|0}$ and $\mathcal{E}_{0|1}$ as $\rho_0$ and $\rho_1$, respectively. The Bloch representation of those states is
\be
\rho_j=\frac 1 2 (\one +\vec{\alpha}_j\cdot \vec{\sigma}),
\ee
for $j\in\{{\rm in}, 0, 1\}$ and Bloch vectors $\vec{\alpha}_j\in \mathbb{R}^3$. 

With these parametrization one can easily observe that our inequality is linear with respect to the parameters $p_j,q_j,|\alpha_j|$, which means the inequality is maximized by pure quantum states, i.e. $|\alpha_j|=1$, and projective measurements. Since we are only considering qubits, the projectors may be of rank $1$ or $2$, and the effects corresponding to rank $2$ projectors are equal to the identity. Only one outcome will be obtained with certainty independent of the state for measurements whose effect is the identity. We name this kind of measurements trivial measurements. For trivial measurements the maximum of $\mathcal{S}_1$ is 6. Hence, the effects which maximize $\mathcal{S}_1$ are of the form
\begin{equation}
\begin{split}
\mathcal{E}_{0|0} &= \frac 1 2 (\one+\vec{c}\cdot\vec{\sigma}),~~~\mathcal{E}_{1|0} = \frac 1 2 (\one-\vec{c}\cdot\vec{\sigma}), \\
\mathcal{E}_{0|1} &= \frac 1 2 (\one+\vec{d}\cdot\vec{\sigma}),~~~\mathcal{E}_{1|1} = \frac 1 2 (\one-\vec{d}\cdot\vec{\sigma}).
\end{split}
\end{equation}
Without loss of generality, we can set
\begin{equation}
\vec{c}=(1,0,0), ~~~\vec{d}=(\cos (-2x_3),\sin (-2x_3),0),
\end{equation}
with $x_3\in[-\pi/2,\pi/2]$. Then $\mathcal{S}_1$ can be rewritten as
\be
3+\frac 1 2 (2\vec{c}\cdot \vec{\alpha}_{\rm in}+\vec{c}\cdot \vec{\alpha}_0+2 \vec{d}\cdot\vec{\alpha}_{\rm in}+\vec{d}\cdot\vec{\alpha}_1)+\frac 1 4(1-\vec{c}\cdot\vec{\alpha}_1)(1-\vec{d}\cdot\vec{\alpha}_0).
\ee
From this equation we can easily find that the maximum is achieved for $\vec{\alpha}_{\rm in}= \frac{\vec{c}+\vec{d}}{|\vec{c}+\vec{d}|}$, and if the vectors $\vec{\alpha}_0$ and $\vec{\alpha}_1$ lie in the plane spanned by the vectors $\vec{c}$ and $\vec{d}$. With this $\mathcal{S}_1$ can be written as 
\be
\mathcal{S}_1=3+\frac 1 4 \left\{2 [\cos (2 x_3+x_0)+\cos (2 x_3+x_1)+4 \cos x_3]+(\cos x_0-1) (\cos x_1-1)\right\}.
\ee
Substituting
\be
x_1=2 \tan ^{-1} a_1,~x_0 = 2 \tan ^{-1} a_0,~x_3 = 2 \tan ^{-1} a_3,
\ee
we obtain for the points where the gradient with respect to the variables $a_0$, $a_1$ and $a_3$ vanishes that
\begin{equation}
\begin{split}
f_1=&2 \left(a_0^2-1\right) \left(a_1^2+1\right) a_3^3-2 \left(a_0^2-1\right) \left(a_1^2+1\right) a_3-2 a_0 \left(4 a_1^2+3\right) a_3^2+a_0 a_3^4+a_0=0,\\
f_2=&2 \left(a_0^2+1\right) \left(a_1^2-1\right) a_3^3-2 \left(a_0^2+1\right) \left(a_1^2-1\right) a_3-2 \left(4 a_0^2+3\right) a_1 a_3^2+a_1 a_3^4+a_1=0,\\
f_3=&-2 a_3^3 \left(a_0^2 \left(3 a_1^2+1\right)+a_1^2-1\right)+2 a_3 \left(a_0^2 \left(a_1^2-1\right)-a_1^2-3\right)+a_3^4 (-(a_0+a_1)) (a_0 a_1+1)\\
&+6 a_3^2 (a_0+a_1) (a_0 a_1+1)-(a_0+a_1) (a_0 a_1+1)=0,
\end{split}
\end{equation}
respectively. We consider then the following four cases:

\noindent
{\bf Case (1).} $a_3=0$. In this case one has $x_3=0$ and can easily see that the bound holds.

\noindent
{\bf Case (2).} $a_3\neq 0$, $a_0\neq a_1$. From the equations
\begin{gather}
f_1-f_2=0,\\
a_0 f_2-a_1 f_1=0,\label{S1}\\
f_3+(1+a_0 a_1)(f_1+f_2)=0,
\end{gather}
we obtain the relation 
\begin{equation}
a_0=\frac 14\left( \frac{4 \left(a_1^2+1\right) a_3}{2 a_1 a_3+a_3^2-1}-2 a_1-a_3+\frac{1}{a_3}\right).
\label{S2}
\end{equation}
Substituting then Eq.~(\ref{S2}) into Eq.~(\ref{S1}), we obtain a polynomial equation of order six in $a_1$ (which is too tedious to present). This equation can be factorized into three quadratic polynomial factors. One of them is ruled out because it leads to $a_0=a_1$, which conflicts with our assumption. Since $\mathcal{S}_1$ is symmetric under the transformation $\{a_0\rightarrow -a_0, a_1 \rightarrow -a_1, a_3\rightarrow -a_3\}$, the two remaining polynomials give the same results. The equation of order six can thus be reduced to a quadratic equation with two different real roots. Due to the symmetry between $a_0$ and $a_1$, we can assign the two different roots respectively to $a_0$ and $a_1$. With this we obtain
\begin{equation}
 f_3= -\frac{\left(a_3^2+1\right)^3 \left(a_3^4+8 a_3-1\right)}{16 (a_3-1)^2 a_3^2}=0.
\end{equation}
This equation has two real roots, one of which will lead to imaginary solutions of $a_0$ and $a_1$. The only possible real root of $a_3$ can be computed analytically, its numerical value is approximately $0.12497$, which leads to $\mathcal{S}_1\approx 5.12402<6$.

\noindent
{\bf Case (3).} $a_3 \neq 0, a_1=a_0$. In this case, the expression of the partial derivatives can be reduced to 
\begin{gather}
   f_1 = f_2=2 \left(a_0^4-1\right) a_3^3-2 \left(a_0^4-1\right) a_3-2 a_0 \left(4 a_0^2+3\right) a_3^2+a_0 a_3^4+a_0, \label{S3}\\
   f_3 = -2 \left(a_0^2+1\right) (a_0 a_3-1) \left(3 a_0 a_3^2-a_0+a_3^3-3 a_3\right).
\end{gather}
We consider different solutions of $f_3=0$ separately. If $a_0 a_3 -1 =0$, then we have
\begin{equation}
f_1= -\frac{\left(a_3^2+1\right)^2 \left(a_3^2+2\right)}{a_3^3}\neq 0.
\end{equation}
If $3 a_0 a_3^2-a_0+a_3^3-3 a_3 = 0,$ that is
\begin{equation}
  a_0 = \frac{3 a_3-a_3^3}{3 a_3^2-1},
  \label{S4}
\end{equation}
one obtains
\begin{equation}
 [1+(-14+t) t)]\{-1+t [8+t (-5+2 t)]\} = 0,
 \label{S5}
\end{equation}
by substituting Eq. (\ref{S4}) into Eq. (\ref{S3}), here $t=a_3^2$. There are three real roots of Eq. (\ref{S5}), namely $a_3=2\pm \sqrt{3}$ or $a_3\approx 0.368671$. In these three cases, $\mathcal{S}_1$ equals to $\frac 1 4 (13\mp 6\sqrt{3})$ or $5.86072$, respectively, which are all strictly smaller than the upper bound $6$.

\noindent
{\bf Case (4).} On the boundary points, where at least one of the parameters $x_1,x_2,x_3$ equals $\pm \pi$, it is easy to prove the validity of desired inequality.
\hfill$\square$

We provide here also an example of a state and measurements violating this inequality, which also shows the nonconvexity of the qubit correlation space in the scenario $L=S=O=2$. It follows from inequality (4) that although the two extreme points
\begin{equation}
\begin{split}
p(00|00)=p(00|01)=p(01|10)=p(00|11)=1,\\
p(00|00)=p(01|01)=p(00|10)=p(00|11)=1,
\end{split}
\end{equation}
are both reachable by a qubit, the uniform mixture of them is not, as it violates inequality (4) in the main text. The algebraic maximum $\mathcal{S}_1=7$ is attained by an extreme point
\be
p(00|00)=p(01|01)=p(01|10)=p(00|11)=1,
\ee
which can be obtained by using a qutrit \cite{HSG18}.

\section{Appendix B: Proof of Observation 2}

We will show in the following that inequality (5) holds. First we denote
\begin{equation}
\mathcal{S}_2=p(0|0,00)+p(0|0,10)+p(1|1,00)+p(1|1,10)+p(0|0,01)+p(1|1,11)+p(1|0,11)+p(0|1,01).
\label{s2}
\end{equation}
Based on the same reasoning as in the proof of Observation 1 (see Appendix A), $\mathcal{S}_2$ is also maximized by projective measurements and pure states. If one of the measurements is trivial, the value of $\mathcal{S}_2$ is no larger than 6. We can write for non-trivial projective measurements, 
\be
\mathcal{E}_{0|j}=\one-|\phi_j\rangle\langle\phi_j|,~j=0,1,
\ee
where $|\phi_j\rangle$ are pure states. Then the expression of $\mathcal{S}_2$ can be rewritten as
\be
\mathcal{S}_2=4-\tr [\rho_1(|\phi_0\rangle\langle\phi_0|+|\phi_1\rangle\langle\phi_1|)]+\tr [\rho_2(|\phi_0\rangle\langle\phi_0|-|\phi_1\rangle\langle\phi_1|)]+\tr [\rho_3(|\phi_1\rangle\langle\phi_1|-|\phi_0\rangle\langle\phi_0|)]+\tr [\rho_4(|\phi_0\rangle\langle\phi_0|+|\phi_1\rangle\langle\phi_1|)],
\ee
where $\rho_1$ is the intermediate state after measurement $0$ is performed and outcome $0$ is produced on the first time step, $\rho_2$ is the intermediate state after measurement $0$ is performed and outcome $1$ is produced, $\rho_3$ corresponds to measurement $1$ and outcome $0$, $\rho_4$ corresponds to measurement $1$ and outcome $1$.

Without loss of generality, we choose $|\phi_0\rangle=|0\rangle$, then parametrize $|\phi_1\rangle=\cos \chi |0\rangle+\sin \chi |1\rangle$. Since $\rho_1$ and $|\phi_0\rangle\langle\phi_0|+|\phi_1\rangle\langle\phi_1|$ are all positive semidefinite matrices, the second term is maximized when $\rho_1$ is the eigenstate corresponding to the smallest eigenvalue of $|\phi_0\rangle\langle\phi_0|+|\phi_1\rangle\langle\phi_1|$. The maximum of the second term can thus be straightforwardly calculated as
 \begin{equation}
 \max_{\rho_1} \{-\tr [\rho_1(|\phi_0\rangle\langle\phi_0|+|\phi_1\rangle\langle\phi_1|)]\}= -\min\{2\cos^2 \left(\frac \chi 2\right),2\sin^2 \left(\frac \chi 2\right) \}.
 \end{equation}
Denote $\rho_2=|\psi\rangle\langle\psi|$, where $|\psi\rangle=\cos \alpha|0\rangle + e^{i\phi}\sin \alpha|1\rangle$, the third term can be written as
\begin{equation}
\begin{split}
\max_{\rho_2}\tr [\rho_2(|\phi_0\rangle\langle\phi_0|-|\phi_1\rangle\langle\phi_1|)]&= \max_{|\psi\rangle}(|\langle 0|\psi\rangle|^2-|\langle \phi_1|\psi \rangle|^2)\\
&=\max_{\alpha,\phi}(\cos^2 \alpha-|\cos \chi \cos \alpha + e^{i\phi}\sin \chi  \sin \alpha|^2)\\
&=\max_\alpha(\cos^2 \alpha-(\cos \chi \cos \alpha + \sin \chi  \sin \alpha)^2)\\
&= \max_\alpha\frac 1 2[ \cos 2\alpha -\cos (2\chi -2 \alpha)]\\
&=\max_\alpha (- \sin 2(\alpha - \frac \chi 2) \sin \chi)\\
&= |\sin \chi|.
\end{split}
\end{equation}
The maximum is achieved when $\sin 2(\alpha - \frac \chi 2)=-{\rm sgn}(\sin \chi)$ and $e^{i\phi}$ equals $1$ or $-1$, here we choose $e^{i\phi}=1$ since the two values lead to the same result due to the maximization over $\alpha$. Using the same method, we find $\max_{\rho_3}\tr (\rho_3(|\phi_1\rangle\langle\phi_1|-|\phi_0\rangle\langle\phi_0|))$ is also $|\sin \chi|$ and $\max_{\rho_4} \tr (\rho_4(|\phi_0\rangle\langle\phi_0|+|\phi_1\rangle\langle\phi_1|))= \max\{2\cos^2 \frac \chi 2,2\sin^2 \frac \chi 2 \}$. Since the maximization is performed over different states for each term, the maximum of $\mathcal{S}_2$ equals the sum of the maximum of all these terms. Therefore the maximal value for qubits is
\begin{equation}
\begin{split}
 \max_{d=2} \mathcal{S}_2&=4 -\min\{2\cos^2 \frac \chi 2,2\sin^2 \frac \chi 2 \}+2\max_\chi  |\sin \chi| + \max\{2\cos^2 \frac \chi 2,2\sin^2 \frac \chi 2 \}\\
 &= 4+\max_\chi (2|\sin \chi|+2|\cos^2\frac{\chi}{2}-\sin^2\frac{\chi}{2}|)\\
 &=4+\max_\chi (2|\sin \chi|+2|\cos \chi|)\\
 &=4+2\sqrt{2}.
\end{split}
\end{equation}

For qutrits, denote $\rho_i=|\psi_i\rangle\langle \psi_i |, i=1,2,3,4$, where $|\psi_i\rangle = \cos \alpha_i\cos \beta_i|0\rangle + e^{i\theta_i}\sin \alpha_i \cos \beta_i|1\rangle+e^{i \phi_i}\sin \beta_i |2\rangle$,  since the second term is non-positive, we can choose $\rho_1=|2\rangle\langle 2|$ to achieve the maximum of the second term. The maxima of the other terms are attained when we choose $\cos \beta_i=1, i =2,3,4$, which reduces to the qubit case. The qutrit bound is the sum of all these terms, which is
\begin{equation}
\begin{split}
 \max_{d=3} \mathcal{S}_2&=4+2\max_\chi  |\sin \chi| + \max\{2\cos^2 \frac \chi 2,2\sin^2 \frac \chi 2 \}\\
&=4+\max_\chi (2|\sin \chi|+2\sin^2 \frac \chi 2)\\
&=5+\sqrt{5}.
\end{split}
\end{equation}
This proves the Observation. \hfill$\square$

If we interprete the Observation in the prepare-and-measure scenario, the qubit bound has already been shown analytically \cite{ANT02}, and the qutrit bound has been obtained numerically \cite{NFA15}. 

An example violating this inequality is presented below, which detects the nonconvexity of the qubit and qutrit correlation space in the simplest scenario $L=S=O=2$.
The extreme points
\begin{equation}
\begin{split}
q_1&: p(00|00)=p(00|01)=p(00|10)=p(01|11)=1,\\
q_2&: p(11|00)=p(10|01)=p(11|10)=p(11|11)=1,
\end{split}
\end{equation}
can be reached by measuring a single qubit. However, the mixture of them, namely $q=\frac{q_1+q_2}{2}$, achieves the maximum $\mathcal{S}_2=8$, thus can not 
be attained even by a qutrit.

\section{Appendix C: Calculations for Example 2}

We consider here the mixture of the extreme point $p(00|00)=p(00|01)=p(00|10)=p(01|11)=1$, which can be realized with measurements on a qubit, and the global white noise $\{p(ab|xy)=1/4, a,b,x,y=0,1\}$.  
We denote the convex weight for the (normalized) global white noise by $\epsilon$, hereafter we show when the noise is sufficiently weak, to be specific, $\epsilon<0.065$, this convex mixture requires at least a qutrit to be realized. For the mixture it holds that $p(0|x)=1-\epsilon/2$, $p(0|0,0y)=p(0|0,10)=p(1|0,11)=\frac{4-3\epsilon}{2(2-\epsilon)}$ and $p(0|1,xy)=p(1|1,xy)=1/2$. 
Note that for the correlation considered, the local-in-time probability distributions (i.e. conditional probabilities $p(b|a,xy)$ for fixed $a,x$) at the second time step can be written either as ${\bf 1}/2$ or $(1-\tilde{\epsilon}) e_i+\tilde{\epsilon}{\bf 1}/2$ with $e_i$ being the tuple $(0,0)$ or $(0,1)$, and $\tilde{\epsilon}=\frac{\epsilon}{2-\epsilon}<0.034$. First we look at the conditional probability distributions, from what we will show in Appendix D, the qubit density matrix that leads to $(1-\tilde{\epsilon}) e_i+\tilde{\epsilon}{\bf 1}/2$ has to be of the form $\rho_i= \mu_i \proj{\Psi_i}+(1-\mu_i)\proj{\Psi_i^\perp}= (1-\mu_i){\bf 1}+(2 \mu_i-1)\proj{\Psi_i}$ with $1/2\leq\mu_i\leq 1$, $\braket{\Psi_i}{\Psi_i^\perp}=0$ and $|\braket{\Psi_1}{\Psi_2}|\leq \frac{2\sqrt{\tilde{\epsilon}}}{1-\tilde{\epsilon}}$. Then we prove in the following that ${\bf 1}/2$ can not be realized by any measurements.

Let us denote as before the effects for measurement $x$ and outcome $a$ by $\mathcal{E}_{a|x}$. We obtain that
\begin{align}\nonumber
\tr (\mathcal{E}_{1|0}\rho_i)&=\frac{\tilde{\epsilon}}{2}\\
&=(1-\mu_i)\tr (\mathcal{E}_{1|0})+(2 \mu_i-1)\tr (\mathcal{E}_{1|0}\proj{\Psi_i})\nonumber\\
&\geq \mu_i \tr (\mathcal{E}_{1|0}\proj{\Psi_i})\nonumber\\
&\geq\frac{1}{2}\ \tr (\mathcal{E}_{1|0}\proj{\Psi_i}),
\end{align}
where we used first that $\tr (\mathcal{E}_{1|0})\geq \tr (\mathcal{E}_{1|0}\proj{\Psi_i})$ due to $\mathcal{E}_{1|0}\geq 0$ and then $\mu_i\geq 1/2$. Hence, we have that
\begin{align}
\tr (\mathcal{E}_{1|0}\proj{\Psi_i})\leq \tilde{\epsilon}.
\end{align}
Therefore, when writing $\mathcal{E}_{1|0}$ in the basis $\{\ket{\Psi_1},\ket{\Psi_1^\perp}\}$, i.e. $\mathcal{E}_{1|0}=\alpha \proj{\Psi_1}+\beta \proj{\Psi_1^\perp}+\gamma \kb{\Psi_1}{\Psi_1^\perp}+\gamma^* \kb{\Psi_1^\perp}{\Psi_1}$, it has to hold that $\alpha\leq \tilde{\epsilon}$.  Note that if $\beta < \alpha$ the largest eigenvalues has to be smaller than $\alpha+\beta<2 \tilde{\epsilon}<1/2$ and therefore it is not possible to attain for this measurement the outcome "1" with probability $1/2$.  

Let us next consider the case $\beta\geq \alpha$.  We expand $\ket{\Psi_2}$ in the same basis, i.e. $\ket{\Psi_2}= \delta\ket{\Psi_1}+\zeta\ket{\Psi_1^{\perp}}$. One obtains that 
\begin{align}\nonumber
\tilde{\epsilon}&\geq \tr (\mathcal{E}_{1|0}\proj{\Psi_2})\\
&=\alpha |\delta|^2+\beta (1-|\delta|^2)+\gamma \zeta \delta ^*+\gamma^* \zeta^* \delta \nonumber\\
&\geq \beta (1-|\delta|^2-2 |\zeta| |\delta|)\nonumber\\
&\geq \beta (1-|\delta|^2-2 |\delta|)\nonumber\\
&= \beta [2-(1+|\delta|)^2]\nonumber\\
&\geq \beta [2-(1+\frac{2\sqrt{\tilde{\epsilon}}}{1-\tilde{\epsilon}})^2],
\end{align}
where  we used for the second line that in this case $\beta\geq |\gamma|$ and  $\alpha |\delta|^2\geq 0$, in the third line $|\zeta|\leq 1$ and in the last line $|\delta|\leq\frac{2\sqrt{\tilde{\epsilon}}}{1-\tilde{\epsilon}}$. Hence, we have that 
\begin{align}
\beta \leq \frac{\tilde{\epsilon}}{ [2-(1+\frac{2\sqrt{\tilde{\epsilon}}}{1-\tilde{\epsilon}})^2]}.
\end{align}
However, for our choice of $\tilde{\epsilon}\leq 0.034$ this implies that $\alpha+\beta<1/2$ and therefore also in this case the probability distribution ${\bf 1}/2$ cannot be realized. Hence, by mixing this extreme point with a small amount of the identity the dimension required to realize the correlation increases.

\section{Appendix D: The proof of Theorem 3}
 
Before proving Theorem 3 let us first provide some useful definition.
Every correlation  in the arrow of time polytope can be decomposed as shown by Eq. (2) in the main text. The conditional probabilities which appear on the right hand side of the equation denote the {\it local-in-time} probabilities at some specific time step while the history is known. Take $p(b|a,xy)$ as an example, denote the {\it local-in-time} probability of getting $b$ as outcome by measuring $y$ at the second time step, with the knowledge of at the first time step measurement $x$ was chosen and outcome $a$ was obtained. With this we define the {\it local-in-time} correlation as probability distributions
\be
\{p(a|hx)=\tr (\rho_h \mathcal{E}_{a|x})\},
\ee
where $h$ stands for the history of measurements and outcomes from the preceding time steps, and $\rho_h$ denotes the intermediate state after the history $h$ took place.  Considering scenarios having $S$ possible measurements, with given and fixed history $h$, we denote local deterministic assignments as tuples $e_i=(a_0,a_1,\ldots,a_{S-1})$, which means local-in-time probability distributions
\begin{equation}
\{p(a_k|hk)=1, k=0,1,\ldots,S-1\}.
\label{tuple}
\end{equation}
With this we can phrase the following lemma which will allow us to prove Theorem 3.

\begin{Lemma}\label{lp}
Let $\frac{\epsilon}{O} {\bf 1} +(1-\epsilon) e_i$ be local-in-time probability distributions in which $\epsilon$ is an arbitrary weight, $\frac{\bf 1}{O}$ is the local normalized identity distribution with $\{p(a|hx)=\frac{1}{O}, \forall a,x\}$ and $e_i$ is a tuple. Moreover,  denote by $\rho_i$ the d-dimensional intermediate state which generates the local-in-time distribution and by $\ket{\gamma_i}$  the eigenstate to its largest eigenvalue. Then it holds for $e_i\neq e_j$ that 
\be
|\langle \gamma_i|\gamma_j\rangle|^2 \leq \frac{d^2\epsilon}{(1-\epsilon)^2}.
\ee
\end{Lemma}
\begin{proof}
Note first that since we consider two different $e_i$ and $e_j$, there exists at least one measurement $x$, for which $e_i$ and $e_j$ give different outcomes, denoted by $a$ and $b$, respectively. Hence, we have
\begin{equation}
\begin{split}
\tr (\mathcal{E}_{a|x} \rho_i)&=\frac{\epsilon}{O},\\
\tr (\mathcal{E}_{b|x} \rho_i)&=1-\epsilon + \frac{\epsilon}{O},\\
\tr (\mathcal{E}_{a|x} \rho_j)&=1-\epsilon + \frac{\epsilon}{O},\\
\tr (\mathcal{E}_{b|x} \rho_j)&=\frac{\epsilon}{O},
\end{split}
\end{equation}
where $\mathcal{E}_{a|x}$ and $\mathcal{E}_{a|x}$ are the corresponding effects. From the above equations one can deduce
\begin{equation}
\begin{split}
\tr(\mathcal{E}_{a|x} (\rho_i-\rho_j)) &= -(1-\epsilon),\\
\tr(\mathcal{E}_{b|x} (\rho_i-\rho_j)) &= (1-\epsilon),
\end{split}
\end{equation}
therefore there exists a decomposition of the identity $\one = P_+ + P_-$ with projectors $P_+$ and $P_-$ satisfying
\begin{equation}
\begin{split}
\tr(P_+ (\rho_i-\rho_j)) &\geq (1-\epsilon),\\
\tr(P_- (\rho_i-\rho_j)) &\leq -(1-\epsilon).
\end{split}
\end{equation}
As $\rho_i$ and $\rho_j$ are both trace one positive semidefinite operators, we get
\begin{equation}
\begin{split}
\tr (P_+ \rho_j) &\leq \epsilon,\\
\tr (P_- \rho_i) &\leq \epsilon.
\end{split}
\label{Seq14}
\end{equation}
The upper bound of the inner product between $\rho_i$ and $\rho_j$ is given by
\begin{equation}
\begin{split}
\tr (\rho_i \rho_j) &= \tr (P_+ \rho_i P_+ \rho_j) + \tr (P_- \rho_i P_- \rho_j)+\tr (P_+ \rho_i P_- \rho_j) + \tr (P_- \rho_i P_+ \rho_j)\\
& \leq 2 \epsilon + \tr (P_- \rho_i P_+  P_+ \rho_j P_-)+\tr(P_+ \rho_i P_-  P_- \rho_j P_+) \leq 4 \epsilon,
\end{split}
\end{equation}
where the first inequality follows from (\ref{Seq14}) and the positivity of $\rho_i,\rho_j$, and the last inequality follows from the Cauchy-Schwarz inequality and the positivity of $\rho_i,\rho_j$. That is
\be
\tr (P_- \rho_i P_+  P_+ \rho_j P_-)\leq\sqrt{\tr(P_-\rho_i P_+ \rho_i )\tr(P_+\rho_j P_-\rho_j)}\leq \epsilon,
\ee
where the second inequality comes from
\begin{equation}
\begin{split}
\tr(P_-\rho_i P_+ \rho_i )&=\tr(P_-\rho_i (\one-P_-) \rho_i )\\
&=\tr (P_- \rho_i^2 )-\tr (P_-\rho_iP_-)\leq \tr (P_-\rho_i).
\end{split}
\end{equation}
Using the spectral decomposition of $\rho_i$ and $\rho_j$, one can rewrite the inequality above as
\begin{equation}
\tr (\rho_i \rho_j) = \sum_{m,n} \mu_m \nu_n |\langle\psi_m|\phi_n\rangle|^2\leq 4\epsilon,
\label{S15}
\end{equation}
where $\{\mu_m, |\psi_m\rangle\}$ and $\{\nu_n, |\phi_n\rangle\}$ are the set of eigenvalues and corresponding eigenvectors of $\rho_i$ and $\rho_j$, respectively. Denoting the largest eigenvalue of $\rho_i$ and $\rho_j$ as $\mu_i=\max_m \mu_m$ and $\nu_j=\max_n \nu_n$, their corresponding eigenvectors as $|\gamma_i\rangle$ and $|\gamma_j\rangle$, we observe that the following inequalities
\begin{equation}
\begin{split}
\mu_i&  \geq \frac{1-\epsilon}{\tr (P_+)},\\
\nu_j&      \geq \frac{1-\epsilon}{\tr (P_-)}=\frac{1-\epsilon}{d-\tr (P_+)}
\end{split}
\end{equation}
always hold. The inequalities are derived directly from inequalities (\ref{Seq14}) and $ P_-+ P_+=\one$. Combining these two inequalities with (\ref{S15}), we can find the overlap between $\gamma_i$ and $\gamma_j$ is upper bounded by
\be
|\langle \gamma_i|\gamma_j\rangle|^2 \leq \frac{4  (d-\tr (P_+))\tr (P_+)}{(1-\epsilon)^2}\leq \frac{d^2\epsilon}{(1-\epsilon)^2},
\ee
which proves the Lemma.
\end{proof}

Lemma \ref{lp} can be also straightforwardly employed in the prepare-and-measure scenario. Moreover, in the proof of Theorem 3 we will use Lemma \ref{lp} in order to show for many cases that the set of correlations has to be non-convex.
\\

\noindent {\bf Theorem 3.} {\it The critical dimension is given by the following formula
\begin{equation}
\mathcal{D}(L,S,O)=\min \{O^S,\frac{(OS)^{L}-1}{OS-1}\}.
\end{equation}
Quantum systems with a dimension that is larger than or equal to the critical dimension generate the correlation polytope $P^L_{S,O}$. Moreover, any correlation space generated by quantum systems with smaller dimension is nonconvex.}
\\

\noindent {\bf Proof.}
 We first consider a specific type of correlations. For them we show that the necessary dimension is given by the critical dimension $\mathcal{D}(L,S,O)$. We will also prove that with the critical dimension it is sufficient to reach all the correlations in the temporal polytope. This allows us to show straightforwardly non-convexity for many instances. We then provide a construction for the remaining cases to prove non-convexity for $d<\mathcal{D}(L,S,O)$. 

We will consider in the following a correlation with all local-in-time probability distributions being of the form given in Lemma \ref{lp}. Additionally, the correlation is chosen to have as many different local-in-time probability distributions as possible. As we will see, in order to produce such a correlation one needs at least a quantum system with dimension $\min \{O^S,\frac{(OS)^{L}-1}{OS-1}\}$. A correlation can have at most $\min \{O^S,\frac{(OS)^{L}-1}{OS-1}\}$ different local probability distributions of this form, since there are $O^S$ possible different tuples $e_i$ and $\frac{(OS)^{L}-1}{OS-1}$ local-in-time distributions. The number of local-in-time distributions of a correlation equals the number of initial and other intermediate states from step $1$ to step $L$. By intermediate state we mean the states that are measured at some point in the sequence. Due to the construction of correlation of the form considered in Lemma \ref{lp}, every outcome would occur in each measurement, which means the the number of intermediate states after one measurement step is $OS$ times more than the intermediate states after the former step. Hence, the 
number of initial state and intermediate states is $\sum_{l=1}^{L}(OS)^{l-1}=\frac{(OS)^{L}-1}{OS-1}$. 

Using Lemma \ref{lp} one obtains that for a $d$-dimensional system to realize such a correlation, the set $\{|\gamma_i\rangle\}$, where $\{|\gamma_i\rangle\}$ are the eigenstates corresponding to the largest eigenvalue of the intermediate state, with cardinality $\min \{O^S,\frac{(OS)^{L}-1}{OS-1}\}$  has to fulfill the pairwise constraints $|\langle \gamma_i|\gamma_j\rangle|^2 \leq \frac{d^2\epsilon}{(1-\epsilon)^2}$. This implies that the states $|\gamma_i\rangle$ are linearly independent if we choose the weight $\epsilon$ to be sufficiently small. We will prove this by contradiction. If $\{|\gamma_i\rangle\}$ is not linearly independent, then $\exists \{\alpha_i \in \mathbb{C}\}, |\Psi\rangle= \sum_i \alpha_i |\gamma_i\rangle=0$. The length of vector $|\Psi\rangle$ can be computed by taking the inner product

\begin{equation}
\begin{split}
\langle\Psi|\Psi\rangle&=\sum_i |\alpha_i|^2 +\sum_{i\neq j} \alpha_i^*\alpha_j \langle \gamma_i |\gamma_j\rangle\\
&\geq \sum_i |\alpha_i|^2 -2\sum_{i< j} |\alpha_i\alpha_j| |\langle \gamma_i |\gamma_j\rangle|\\
&\geq \sum_i |\alpha_i|^2 -2\sum_{i< j} \frac{d\sqrt{\epsilon}}{1-\epsilon} |\alpha_i\alpha_j|
\end{split}
\label{independent}
\end{equation}

From Eq. (\ref{independent}) we can see for $d$-dimensional quantum systems, if $\epsilon< \frac{1}{2}(1+d^2(d-1)^2-d(d-1)\sqrt{4+d^2(d-1)^2})$, then we have
\begin{equation*}
\begin{split}
\langle\Psi|\Psi\rangle&\geq \sum_i |\alpha_i|^2 -2\sum_{i< j} \frac{d\sqrt{\epsilon}}{1-\epsilon} |\alpha_i\alpha_j|\\
&=\frac{1}{d-1}\sum_{i<j}(|a_i|^2-\frac{2d(d-1)\sqrt{\epsilon}}{1-\epsilon} |\alpha_i\alpha_j|+|a_j|^2)\\
&>\frac{1}{d-1}\sum_{i<j}(|a_i|^2-2|a_i a_j|+|a_j|^2)\\
&=\frac{1}{d-1}\sum_{i<j}(|a_i|+|a_j|)^2> 0,
\end{split}
\end{equation*} 
which contradicts the assumption of $|\Psi\rangle =0$ and therefore the vectors have to be linearly independent. However, if $d<\min \{O^S,\frac{(OS)^{L}-1}{OS-1}\}$  there cannot exist $\min \{O^S,\frac{(OS)^{L}-1}{OS-1}\}$ linearly independent vectors in the Hilbert space. Hence, in this case such a correlation cannot be realized.  
With this we have shown that the cardinality of $\{|\gamma_i\rangle\}$, which equals to the cardinality of the deterministic tuple set $\min \{O^S,\frac{(OS)^{L}-1}{OS-1}\}$, is the dimension necessary  to realize every point in the correlation polytope. 

On the other hand, if the dimension of the underlying quantum system is $\min \{O^S,\frac{(OS)^{L}-1}{OS-1}\}$, we can use it to construct protocols that are able to realize an arbitrary point in $P^L_{S,O}$. If the quantum system has dimension $d=O^S\leq\frac{(OS)^{L}-1}{OS-1}$, we can have $O^S$ pure orthogonal quantum states, denoted as $\{|0\rangle, |1\rangle, \ldots, |O^S-1\rangle\}$. Assigning each of the $O^S$ deterministic tuples to one pure state, we can construct the $S$ measurements such that measuring these measurements on state $|i\rangle$ can produce the corresponding tuple. Explicitly, the measurements are constructed as projective measurements with effects $\mathcal{E}_{r|s} = \sum_{\{i:p_i(r|s)=1\}} |i\rangle\langle i|$, here $p_i(r|s)=1$ means that the outcome $r$ will be produced deterministically while the $s$-th measurement is performed on the state $|i\rangle$. Given an arbitrary correlation, we can calculate the local probability distributions at every time step and decompose them as convex combinations of deterministic tuples. By tuning every intermediate state to be a mixture of the according orthogonal quantum states, with the weight of each state equals to the weight of its corresponding deterministic tuple, we can realize all the local probability distributions and thus the correlation itself.  

If the quantum system has dimension $d=\frac{(OS)^{L}-1}{OS-1}\leq O^S$, we can set all the intermediate state as orthogonal pure states, and design the effect of POVM according to the correlations we want to achieve. Taking $L=2$ case as an example, we set the initial state to be $|\psi_0\rangle$, and the state we get after obtaining outcome $a$ for measurement $x$ in the fist time step as $|\psi_{a|x}\rangle.$ With a $(OS+1)$-dimensional quantum system, $\{|\psi_0\rangle,|\psi_{a|x}\rangle\}$ can be chosen as a orthogonal vector set. Any correlation $\{p(ab|xy)\}$ can then be realized by a set of measurements $\{\mathcal{M}_s, s=0,1,\ldots, S-1\}$ whose effects are  $\{\mathcal{E}_{r|s} = p(r|s) |\psi_0\rangle\langle \psi_0| + \sum_{a,x} p(r|axs) |\psi_{a|x}\rangle\langle \psi_{a|x}|,r=0,1,\ldots,O-1\}$.

The remaining part is to prove that the correlation space produced by a quantum system with dimension $d<\mathcal{D}(L,S,O)$ is nonconvex.

We divide the situation into two cases, either one could still realize all the extreme points with the $d$-dimensional system, or one could not.  From the preceding proof, it is obvious that temporal correlation spaces generated by quantum systems with dimension strictly smaller than $\mathcal{D}(L,S,O)$ but still able to realize all the extreme points of $P^L_{S,O}$, are nonconvex. 

Extreme points of $P^L_{S,O}$ are deterministic assignments, the local-in-time probablity distributions of them are tuples, as defined in Eq.~(\ref{tuple}). With a $d$-dimensional quantum system that cannot reach all the extreme points, one can reach any extreme point which has at most $d$ different tuples, while not being able to produce extreme points with $d+1$ different tuples (see also \cite{SBG19}).  Based on this, we construct two extreme points which can be realized with a $d$-dimensional system, but the mixture of them can only be realized by $(d+1)$-dimensional systems. The first extreme point gives result "0" for all the measurements, i.e., the tuple $(0,0,\ldots,0)$ is generated as local-in-time probability distribution in the first time step, then generates exactly $d-1$ different tuples which are not identical with $(0,0,\ldots,0)$ or $(1,1,\ldots,1)$ in the following time steps, and all the remaining local-in-time probability distributions are the tuple $(0,0,\ldots,0)$. The second point gives the tuple $(1,1,\ldots,1)$ whenever the tuple $(0,0,\ldots,0)$ is generated in the first extreme point, while its other local-in-time probability distributions being identical with the first point. This construction always exists for any $d$-dimensional quantum system that can not reach all the extreme points, since every extreme point of $P^L_{S,O}$ has at least $d+1$ local probability distributions, and there exists extreme points with at least $d+1$ different tuples, otherwise all of them can be realized by a $d$-dimensional system. 

 Both the points we consider can be realized by $d$-dimensional quantum systems. The uniform mixture of them, however, needs a quantum system with at least dimension $d+1$ to realize. This can be conceived as follows: the uniform mixture of them has to realize $d+1$ different deterministic tuples as local-in-time probability distributions, the $d+1$ intermediate states that give the tuples are orthogonal to each other (see, e.g. \cite{NiC10} or (\ref{S15}), with $\epsilon=0$). Therefore we need at least $(d+1)$-dimensional quantum system to realize the mixture, which finishes the proof.
 \hfill \qed

\section{Appendix E: Detailed description of the numerical algorithm}

Consider any given polynomial $p(X_1,X_2,\ldots,X_n)$ where the $X_i$ are the involved local probabilities of the form $p(a|x)$ or $p(b|a,xy)$. Since every maximization problem can be converted into a minimization problem, we only present the method for finding the minimum of such a polynomial. To find the minimum of $p(X_1,X_2,\ldots,X_n)$ for a $d$-dimensional quantum system, we can first choose a random number $q$, and check whether $p(X_1,X_2,\ldots,X_n)$ can achieve a value smaller than $q$ with correlations obtained from measuring a $d$-dimensional system. We illustrate this using the $d=2$ case as an example. For a correlation that can be produced by a qubit, its corresponding $(X_1,X_2,\ldots,X_n)$ has a quantum representation $X_i=\tr (\rho_i M_i)$, with $\rho_i$ being the initial or intermediate states and $M_i$ the measurement effects. By assumption, the polynomial is minimized by a correlation with pure states $\rho_i = |\psi_i\rangle\langle\psi_i|$ and projective measurement effects $M_i = |\phi_
i\rangle\langle\phi_i|$. For this correlation we can construct a $2\times 2n$ matrix
\begin{equation}
\Gamma =  \begin{pmatrix}
   |\psi_1\rangle, \ldots,    |\psi_n\rangle,   |\phi_1\rangle,  \ldots,  |\phi_n\rangle|
\end{pmatrix}.
\label{eq5}
\end{equation}
Then, the matrix $\Gamma^\dagger \Gamma$ is a $2n \times 2n$ positive semi-definite matrix with all diagonal entries equal to $1$ and rank 2. Every $X_i =\tr (\rho_i M_i)=|\langle\psi_i|\phi_i\rangle|^2$ is the absolute square of a certain entry. If the minimum of $p(X_1,X_2,\ldots,X_n)$ is smaller than a number $q$, then there should exist a common object in the following two sets of $2n \times 2n$ matrices:

\noindent
($M_1$) Rank two positive semi-definite matrices.

\noindent
($M_2$) Hermitian matrices  with the main diagonal $(1,1,\ldots,1)$, whose entries corresponding to $\{X_i\}$ satisfy the inequality
\begin{equation}
p(X_1,X_2,\ldots,X_n)\leq q.
\end{equation}

To examine the existence of such a matrix, one can iterate between these two sets, as shown in Fig. \ref{fig}. For a given object in $M_1$ one can find the closest object
in $M_2$ and vice versa. Each step of the iteration is analytical. A common object 
exists if the iteration converges, the converse is however not true. 

In more detail, we can first take a random matrix $H_1$ in the first set, and find the closest point on the border of the second set, i.e., find a Hermitian matrix $H_2\in M_2$ which minimizes $||H_1-H_2||_F$, where $||A||_F=\sqrt{\tr AA^\dagger}$ is the Frobenius norm. This can be done using the method of Lagrange multipliers. Since $H_2$ is a square Hermitian matrix, it can be written as $H_2=U_2 V U_2^\dagger$, where $U_2$ is a unitary matrix and $V_2=(v_1,v_2,\ldots,v_{2n})$ is a diagonal matrix with $v_1\geq v_2 \geq \ldots \geq v_{2n}$. Denoting $V_3=(v_1,v_2,0,\ldots,0)$, the matrix closest to $H_2$ in the Frobenius norm in the first set is then $H_3=U_2 V_3 U_2^\dagger$ \cite{HoJ13}. If one matrix is found to be in both sets using this iteration, then $q$ is larger than the minimal value.

The minimum lies in the interval $[q_{min},q_{max}]$, where $q_{min}$ and $q_{max}$ are the algebraic minimum and maximum of $p(X_1,X_2,\ldots,X_n)$, respectively, it can be obtained via binary search. First we examine whether or not a common object of the sets $M_1$ and $M_2$ with $q=(q_{min}+q_{max})/2$ exists. If yes, we keep investigating the middle point of a new interval $[q_{min},(q_{min}+q_{max})/2]$, otherwise we test the middle point of interval $[(q_{min}+q_{max})/2,q_{max}]$, until the length of the interval is smaller than a preset accuracy.

This method can be generalized to quantum systems with $d>2$, where the rank of non-trivial projective measurement effects can have different values. We can calculate all the lower bounds according to  possible measurement effect ranks and then the smallest one is the lower bound of $p(X_1,X_2,\ldots,X_n)$. If a specific measurement effect $M_i$ is of rank $r$, we can choose a set of its eigenvectors ${|\phi_i^1\rangle,|\phi_i^2\rangle,\ldots,|\phi_i^r\rangle}$ and construct $\Gamma$ as

\begin{equation}
\Gamma =  \begin{pmatrix}
   |\psi_1\rangle, \ldots,    |\psi_n\rangle, |\phi_1\rangle, \ldots, |\phi_i^1\rangle, \ldots, |\phi_i^r\rangle,\ldots,  |\phi_n\rangle |
\end{pmatrix}.
\end{equation}
This construction imposes more linear constraints on the second set of matrices, while the diagonal block corresponding to a rank $r$ measurement effect becomes a $r\times r$ identity.

\begin{figure}
\centering     
\includegraphics[width=0.4\columnwidth]{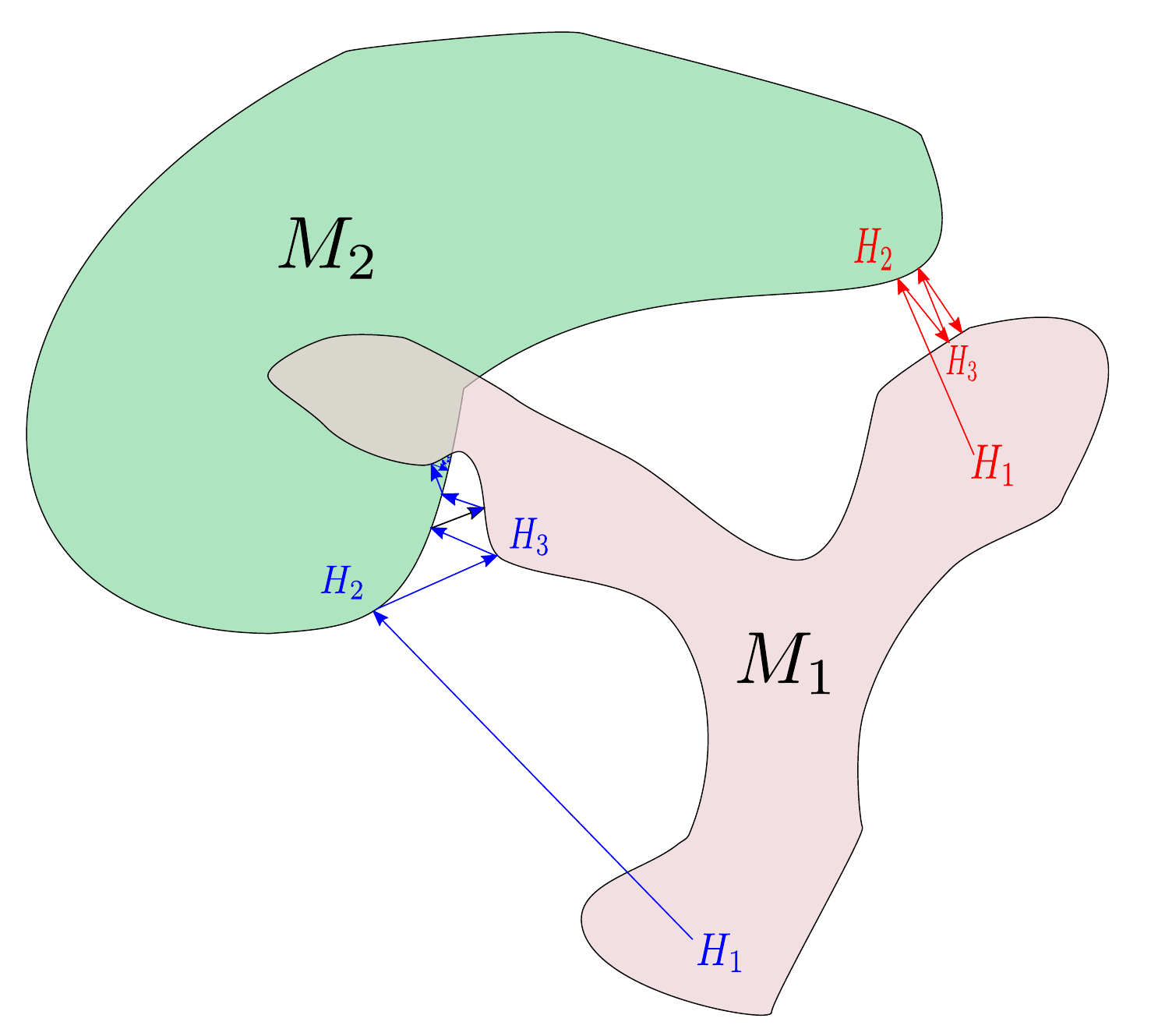}
\caption{Schematic illustration of the algorithm. The blue arrows demonstrate steps of the algorithm which converge to a common point. Note that the algorithm does not necessarily converge even if the two sets have common elements. The path in red exemplifies this case.} 
\label{fig}
\end{figure}

\section{Appendix F: An application of the numerical algorithm} 

Using the normalization $p(0|a,xy)=1-p(1|a,xy)$, the polynomial on the left hand side of inequality (5) in the main text can be rewritten as 

\begin{equation}
\begin{split}
&8-[p(1|0,00)+p(1|0,01)+p(1|0,10)+p(1|1,01)\\
&+p(0|1,10)+p(0|1,11)+p(0|0,11)+p(0|1,00)].
\end{split}
\label{eq3}
\end{equation}
Then the problem of finding the upper bound of the inequality is equivalent to minimizing $p(1|0,00)+p(1|0,01)+p(1|0,10)+p(1|1,01)+p(0|1,10)+p(0|1,11)+p(0|0,11)+p(0|1,00)$. The new polynomial involves four intermediate states $\rho_{i|s}$, with $i,s=0,1$, and four measurement effects $\mathcal{E}_{r|s}$, with $r,s=0,1$. Since we are sure that the minimum lies in interval $[0,8]$, we set $q=4$, the matrix $G=(G_{ij})$ we are looking for is a $8\times 8$ positive semi-definite matrix with all diagonal entries equal to $1$ and the function $p(G)=|G_{16}|^2+|G_{18}|^2+|G_{36}|^2+|G_{28}|^2+|G_{45}|^2+|G_{47}|^2+|G_{37}|^2+|G_{25}|^2\leq 4$. If such a matrix is found, then we know then the minimum of $p(1|0,00)+p(1|0,01)+p(1|0,10)+p(1|1,01)+p(0|1,10)+p(0|1,11)+p(0|0,11)+p(0|1,00)$ is in the interval $[0, 4]$. In order to find such a matrix, we iterate between the following two matrix sets:

(1) Rank two positive semidefinite $8\times 8$ matrices.

(2) Hermitian matrices of the form 
\begin{equation}
G =  \begin{pmatrix}
1 &  &    &   &   & G_{16} &  & G_{18} \\
  &  1 &&&G_{25}&&& G_{28}\\
  
  &&1&&&G_{36}&G_{37}&\\
  &&&1&G_{45}&&G_{47}&\\
  
  &&&&1&&&\\
  &&&&&1&&\\
  
  &&&&&&1&\\
  &&&&&&&1\\
\end{pmatrix}.
\label{eq8}
\end{equation}
Note that gaps in the above matrix represent entries that are not specified beyond the hermiticity condition. Further, the entries of $G$ fulfill the condition $p(G) \leq q.$

For any rank two positive semidefinite $8\times 8$ matrix $H=(H_{ij})$, assume the matrix closest to $H$ on the boundary the second set is the Hermitian matrix $H'=(H'_{ij})$ of the form specified in Eq. (\ref{eq8}). By constructing the Lagrangian function $p(H'-H)-\lambda[p(H')-q]$,
in which $\lambda$ is the Langrange multiplier, and solving the equations $\nabla p(H'-H) = \lambda \nabla p(H')$ and $p(H')=q$, we obtain the explicit expression $H'_{ij} = H_{ij}/p(H), \forall i \neq j$. If the iteration converges, then $q$ is larger than the minimum we are looking for. In this case we update our knowledge and search in the new interval $[q_{min},q]$. Using this method we can find the upper bound of inequality (5) for a qubit numerically.

Due to the symmetry of the polynomial, we can choose $\mathcal{E}_{0|s}, s=0,1$ to be rank one and $\mathcal{E}_{1|s}, s=0,1$ to be rank two for the qutrit case. The first set of matrices is then consisting matrices of the form

\begin{equation}
G =  \begin{pmatrix}
1 &     &    &   &       & G_{16}^{(1)} &G_{16}^{(2)}  & & G_{18}^{(1)}& G_{18}^{(2)} \\
   &  1 &    &   &G_{25} &           &           &  &G_{28}^{(1)}&G_{28}^{(2)}\\ 
  
  &&1&&&G_{36}^{(1)}&G_{36}^{(2)}&G_{37}& &\\
  &&&1&G_{45}&&&G_{47}&&\\
  
  &&&&1&&&\\
  &&&&&1&0&\\
  &&&&&&1&\\
  &&&&&&&1&\\
  &&&&&&&&1&0\\
  &&&&&&&&&1\\
\end{pmatrix},
\end{equation}
fulfilling $p(G) \leq q.$
Here $G_{ij}^{(1)}$ and $G_{ij}^{(2)}$ denote the inner products of the state vectors and two orthonormal eigenvectors of measurement effects $\mathcal{E}_{1|s}$, respectively, and we define $|G_{ij}|^2=|G_{ij}^{(1)}|^2+|G_{ij}^{(2)}|^2, \text{for}~ij\in \{16,18,28,36\}$. Other parts of the algorithm are the same as for a qubit.

\vspace{0.9cm}
\twocolumngrid
\vspace{0.9cm}

\end{document}